\documentclass[a4paper, 11pt]{article}
\usepackage{amsmath}
\usepackage{amssymb}
\usepackage{amsthm}

\newtheorem{theorem}{Theorem}[section]
\newtheorem{proposition}[theorem]{Proposition}
\theoremstyle{definition}
\newtheorem{definition}[theorem]{Definition}

\newcommand{\tr}{\operatorname{tr}}

\begin{document}

\title{Detailed balance and entanglement}
\author{Rocco Duvenhage and Machiel Snyman\\
Department of Physics, University of Pretoria, Pretoria 0002 \\South Africa }
\date{18 March 2015}
\maketitle

\begin{abstract}
We study a connection between quantum detailed balance, which is a concept
of importance in statistical mechanics, and entanglement. We also explore
how this connection fits into thermofield dynamics.

\end{abstract}

PACS numbers: 03.65.Yz, 05.30.-d, 03.65.Aa, 03.65.Ud

\section{Introduction}

Entanglement is a central aspect of quantum physics. It is for example
well established as a core concept in the broad field of quantum information
\cite{H}. It has also become clear that it has important applications in other
areas of physics. One such area where much work has been done recently is
statistical mechanics. See for example the book \cite{GMM} and the reviews
\cite{AmE}, as well as the papers \cite{GL} for various ideas that have
been explored in this connection. It is therefore of interest to explore
further general connections between entanglement and statistical
mechanics. In particular in this paper we consider a connection to detailed
balance.

Detailed balance is a form of microscopic reversibility and is intimately
related to equilibrium. Quantum versions of detailed balance for open systems,
which is what we are interested in this paper, have been studied
for many years, one of the earliest papers being \cite{Ag}.
Other early work includes \cite{CWA, KFGV}. This line of research continues in
the present day as seen for example in \cite{FU2}, and includes studies of
related aspects of dynamics, like mixing times, \cite{TeTe}. There are various
approaches to quantum detailed balance with varying degrees of generality, as
illustrated by the mentioned papers.

Connections between detailed balance and entangled states have in fact already
been exploited in \cite{FR, FR2, BQ2} with regards to entropy production for
quantum Markov semigroups (also see \cite{BQ} for related work). Here our goal
is to study this connection itself more explicitly, in particular how it arises
as well as one instance of how it fits into other parts of physics,
specifically the area known as thermofield dynamics.

We only consider systems with finite dimensional Hilbert space in this paper. The
relevant concepts regarding entanglement, in particular a convenient
representation of purifications, are presented in Section \ref{a2}.
A heuristic motivation as to why one might in general expect a connection between
detailed balance and entanglement is presented in Section \ref{mot}. Two
definitions of quantum detailed balance, one of which was also considered in
\cite{FR, FR2, BQ2}, are discussed in Section \ref{a3}.
The characterization of these forms of detailed balance in terms of a
certain entangled state is then described in Section \ref{a4}, and proved in
Section \ref{bewyse}. In Section \ref{tvd} we show how these results fit naturally
into thermofield dynamics. Further general remarks are made in
Section \ref{a6}.

\section{Entanglement}

\label{a2}

Here we set up a representation of the purification of a state, which will be
convenient when we study the connection between detailed balance and
entanglement in Section \ref{a4}. At the same time we introduce some notation
that will be used in the rest of the paper.

Consider a quantum system with $n\geq2$ dimensional Hilbert space whose state
is given by the density matrix $\rho$. The expectation value of an observable
$A$ of the system is therefore given by
\[
\left\langle A\right\rangle =\tr(\rho A).
\]
For mathematical convenience we define this functional $\left\langle
\cdot\right\rangle $ on the whole of the algebra $M_{n}$ of $n\times n$
complex matrices, rather than just on the self-adjoint matrices. Note that
$\rho$ can be recovered from $\left\langle \cdot\right\rangle $ so we may view
$\left\langle \cdot\right\rangle $ as a representation of the system's state.
Denoting the Hilbert-Schmidt inner product by $(\cdot|\cdot)$, we have
\[
\left\langle A\right\rangle =\tr(r^{\dagger}Ar)=(r|Ar)
\]
for any $n\times n$ matrix $r$ such that $\rho=rr^{\dagger}$. Note that such
matrices $r$ exist exactly because $\rho\geq0$.

We introduce a faithful representation $\pi$ of the tensor product
$M_{n}\otimes M_{n}$ on the space $M_n$ by
\begin{equation}
\pi(A\otimes B)X=AXB^{\intercal} \label{voorst}
\end{equation}
where $B^{\intercal}$ is the transpose of the matrix $B$, while $X$ is any
element of the representation space $M_n$. Note that this representation 
depends on the basis we are using, because of the transpose. Keep in mind that 
$\pi$ is well defined on the whole of $M_{n}\otimes M_{n}$ because of the 
universal property of tensor products. We can view $\pi$ as faithfully 
representing $M_{n}\otimes M_{n}$ on the Hilbert space $M_{n}$ with the 
Hilbert-Schmidt norm, and in particular this Hilbert space can be taken as the 
Hilbert space of two copies of the system together, which we call the 2-system. 
So, if $X$ in Eq. (\ref{voorst}) is a normalized element of the Hilbert space 
$M_n$, then it represents a pure state of the 2-system. A way to see all this 
easily is to represent a pure state of the first system as a column vector 
$\psi$ in the $n$ dimensional Hilbert space, but to take the transpose of a 
pure state $\phi$ of the second system to get a row vector $\phi^{\intercal}$, 
in which case the elementary tensor $\psi\otimes\phi$ can be written as the 
matrix product
\[
\psi\phi^{\intercal}=\left[
\begin{array}
[c]{c}
\psi_{1}\\
\vdots\\
\psi_{n}
\end{array}
\right]  \left[
\begin{array}
[c]{ccc}
\phi_{1} & \cdots & \phi_{n}
\end{array}
\right]  ,
\]
since this is simply the Kronecker product of the two pure states (in terms of
their components $\psi_{1},...,\psi_{n}$ and $\phi_{1},...,\phi_{n}$
respectively), represented as an $n\times n$ matrix. The general pure state
$X$ of the 2-system is simply a linear combination of such elementary tensors.
In this representation it is clear that when $A\otimes B$ acts on $\psi
\otimes\phi$, i.e. when $A$ acts on $\psi$ and $B$ on $\phi$, then it is
represented by
\[
(A\psi)(B\phi)^{\intercal}=A\psi\phi^{\intercal}B^{\intercal}
\]
which extends linearly to Eq. (\ref{voorst}) for general pure states $X$ of the
2-system.

Using this representation and viewing $r$ above as a pure state of the 2-system, 
we define the corresponding expectation functional $\omega$ on $M_{n}\otimes
M_{n}$ by
\begin{equation}
\omega_{r}(A\otimes B)=(r|\pi\lbrack A\otimes B]r)
=\tr(r^{\dagger}ArB^{\intercal}). \label{omegar}
\end{equation}
We use the notation $\omega_{r}$ rather than, say,
$\left\langle\cdot\right\rangle _{r}$,
to distinguish it more clearly from $\left\langle\cdot\right\rangle $,
especially later on when we drop the subscript $r$. We can therefore view
$\omega_{r}$ as a pure state of the 2-system (represented as an expectation
functional), and since in terms of the $n\times n$ identity matrix $I$ we
clearly have
\[
\omega_{r}(A\otimes I)=\left\langle A\right\rangle
\]
where the left hand side corresponds to taking a partial trace, we see that
$\omega_{r}$ is a purification of $\left\langle \cdot\right\rangle $, i.e. the
state $r$ in the Hilbert space $M_{n}$ is a purification of $\rho$. (At this
stage we have not assumed that $\rho$ is necessarily mixed, but we will do so
later.) This construction of $\omega_{r}$ is closely related to constructions
used in \cite{WFD}, but the specific representation Eq. (\ref{voorst}) is
different, and in the mentioned references the tensor product of two slightly
different algebras are taken instead of two copies of the same algebra $M_{n}$
as in our case.

As already mentioned, $\omega_{r}$ depends on the basis in which are working,
but the fact that we allow any $r$ such that $\rho=rr^{\dagger}$, in effect
compensates for this, as we now explain. If we were to change the basis we are
working in by a unitary transformation $V$, i.e. $A$, $B$ and $\rho$ are
replaced by $V^{\dagger}AV$, $V^{\dagger}BV$ and $V^{\dagger}\rho V$
respectively, so in particular we would use $r$ such that
$rr^{\dagger}=V^{\dagger}\rho V$, then the definition of $\omega_{r}$ would
change to
\[
\omega_{r}(A\otimes B)
=\tr[r^{\dagger}V^{\dagger}AVr(V^{\dagger}BV)^{\intercal}]
=\tr(r_{V}^{\dagger}Ar_{V}B^{\intercal})
\]
where $r_{V}=VrV^{\intercal}$ which clearly satisfies
$r_{V}r_{V}^{\dagger}=\rho$, so we are back to the original definition,
expressed in the original basis, by making a different choice of $r$, namely
$r_{V}$.

Without loss of generality we can therefore assume that in Eq. (\ref{omegar})
we are working in a basis in which $\rho$ is diagonal, which is indeed what we
do in the rest of the paper. Furthermore, it is easily shown that the most
general form for such $r$ is $r=\rho^{1/2}W$ where $W$ is any $n\times n$
unitary matrix.

In the rest of this paper we focus on the choice $r=\rho^{1/2}$ in which case
we denote $\omega_{r}$ simply by $\omega$, i.e.
\begin{equation}
\omega(A\otimes B)=\tr(\rho^{1/2}A\rho^{1/2}B^{\intercal})
\label{omega}
\end{equation}
with $\rho$ diagonal. The reason for this is that it ensures that
\[
\omega(I\otimes B)=\left\langle B\right\rangle ,
\]
i.e. both copies of the system are in the same state $\rho$. More generally
this can be ensured by requiring not only $rr^{\dagger}=\rho$, but also
$r^{\dagger}r=\rho$, since $\rho^{\intercal}=\rho$, however
$r=\rho^{1/2}$ is the simplest option.

To summarize, $\omega$ is a pure state of the 2-system whose reduced states to
both systems are given by $\left\langle \cdot\right\rangle $, i.e. by $\rho$,
and since in statistical mechanics we are particularly interested in cases
where $\rho$ is not pure, it follows then that $\omega$ is an entangled state.

Throughout the rest of the paper we in fact assume that $\rho$ is invertible,
i.e. all its eigenvalues are strictly positive. In particular $\rho$ is not a
pure state, and therefore the pure state $\omega$ is entangled.

\section{Detailed balance and correlated states}

\label{mot}

Next we present a somewhat heuristic discussion of why a connection between
quantum detailed balance and entanglement can be expected. In order to do
this we start with detailed balance for a classical Markov chain and show
how it can be expressed in terms of a correlated state of two copies of the
system in question, where both of its marginals are the original state of the
system.

Recall that if we have a probability distribution $p_{1},...,p_{n}$ over a
finite set $F$ of $n$ elements, then a Markov chain satisfying detailed
balance is described by transition probabilities $\gamma _{jk}$ satisfying
\begin{equation*}
p_{j}\gamma _{jk}=p_{k}\gamma _{kj}
\end{equation*}
for all $j,k=1,...,n$, which simply says that the probability to make a
transition from one pure state to another is the same as the opposite
transition. Denoting the observable algebra of functions on the $n$-point
set $F$ by $K=\mathbb{C}^{n}$, we can express the probability distribution
$p_{1},...,p_{n}$ by a normalized positive linear functional (a \emph{state}) $\mu $
on $K$ given by
\begin{equation*}
\mu (f)=pf
\end{equation*}
where $f\in K$ is viewed as a column matrix and
$p=[
\begin{array}{ccc}
p_{1} & \cdots & p_{n}
\end{array}
]$ is a row matrix. Now we consider two copies of the algebra $K$, namely
the tensor product algebra $K\otimes K$ and define a state $\varphi $ on it
by
\begin{equation*}
\varphi =\mu \circ \delta
\end{equation*}
where $\delta :K\otimes K\rightarrow K$ is given by componentwise
multiplication, i.e. $\delta (f\otimes g)=fg$ where $fg$ is the product in
the algebra $K$, defined to have the components $f_{j}g_{j}$ if $f$ and $g$
have components $f_{j}$ and $g_{j}$ respectively. Note that $\delta $ is
well-defined because of the universal property of the tensor product. It is
clear that $\varphi $ corresponds to the probability distribution
$p_{1},...,p_{n}$ over the ``diagonal'' of the set $F\times F$ and is
therefore a correlated state unless all but one of the probabilities are
zero. Note that analogous to the entangled state $\omega $ from the previous
section, the marginals of $\varphi $ are simply the state $\mu $ of the
single system we started with, namely
\begin{equation*}
\varphi (f\otimes 1)=\mu (f)\text{ \ \ and \ \ }\varphi (1\otimes g)=\mu (g)
\end{equation*}
where the $1$ here denotes the function which is identically $1$ on $F$,
i.e. the column consisting only of $1$'s.

Denoting the transition matrix by $\Gamma =(\gamma _{jk})$, the
time-evolution on $K$ is given by $f\mapsto \Gamma f$, and using the
detailed balance condition above it follows that
\begin{equation*}
\varphi ((\Gamma f)\otimes g)=\sum_{j=1}^{n}\sum_{k=1}^{n}p_{j}\gamma
_{jk}f_{k}g_{j}=\sum_{j=1}^{n}\sum_{k=1}^{n}p_{k}\gamma
_{kj}f_{k}g_{j}=\varphi (f\otimes (\Gamma g))
\end{equation*}
and conversely, if
\begin{equation}
\varphi \lbrack (\Gamma f)\otimes g]=\varphi \lbrack f\otimes (\Gamma g)]
\label{klasfyn}
\end{equation}
holds for all $f,g\in K$, then the detailed balance condition $p_{j}\gamma
_{jk}=p_{k}\gamma _{kj}$ follows easily. So the detailed balance condition
of a system, which says that a transition and its opposite are equally
likely, can be reinterpreted in terms of two copies of the system by saying
that in the correlated state $\varphi $ time-evolution of only the first
copy of the system is equivalent to time-evolution of only the second copy
of the system, i.e. the two systems' time-evolutions are ``balanced'' in
this sense. It is clear from the derivation of Eq. (\ref{klasfyn}) from
detailed balance that the fact that the transition probability $p_{j}\gamma
_{jk}$ is equal to the opposite transition's probability $p_{k}\gamma _{kj}$
in the first system, is translated directly to time-evolution of the second
system. A potentially useful way of thinking about this may be that the
first system is going back in time relative to the second, in the right hand
side of Eq. (\ref{klasfyn}).

The above discussion makes it plausible that also in the quantum case
detailed balance of a system should be related to a correlated state of two
copies of the system. Our next step is to explore this in more detail to
motivate the connection between detailed balance and entanglement. More
precisely, if we attempt to express quantum detailed balance in the form of Eq.
(\ref{klasfyn}), the question is which state of two copies of the quantum
system should be used in place of $\varphi $.

A most direct adaptation of the state $\varphi $ to the quantum case from
the previous section is to consider the following density matrix for two
copies of the quantum system, where as for the classical case above we assign
the probabilities only to pairs consisting of two copies of the same pure
state (i.e. a probability distribution over a ``diagonal'' of 2-system pure
states):
\begin{equation*}
\rho ^{(2)}=\sum_{j=1}^{n}\rho _{j}\left| e_{j}\otimes e_{j}\right\rangle
\left\langle e_{j}\otimes e_{j}\right|
\end{equation*}
where we are working in a basis in which $\rho $ from the previous section
is diagonal, say
\begin{equation}
\rho =\left[
\begin{array}{ccc}
\rho _{1} &  &  \\
& \ddots &  \\
&  & \rho _{n}
\end{array}
\right] ,  \label{dm}
\end{equation}
and with $e_{j}$ the column matrix with $1$ in the $j$ 'th position and $0$
elsewhere for $j=1,...,n$, to give the pairs of states $e_{j}\otimes e_{j}$
referred to above. Then it is easily verified that if we define a
state $\theta $ on $M_{n}\otimes M_{n}$ by
\begin{equation*}
\theta (C)=\tr(\rho ^{(2)}C)
\end{equation*}
for all $C\in M_{n}\otimes M_{n}$, then $\theta (A\otimes I)=\left\langle
A\right\rangle $ and $\theta (I\otimes B)=\left\langle B\right\rangle $ as
required to correspond to the classical case above, and it is also clear
that $\theta $ is a correlated state (as long as more than one of the $\rho
_{j}$ are non-zero) although it contains no entanglement, i.e. the
correlations in $\theta $ are purely classical. One could now try to define
quantum detailed balance, for some time-evolution of the system, in terms of
$\theta $ by using a similar condition as in Eq. (\ref{klasfyn}).

However the question is whether $\theta$ is sufficiently correlated to
produce a good analogue of the classical case. So let us heuristically
compare $\theta $ with $\varphi $ in terms of how correlated they are. Let
us assume that $\rho _{j}\neq 0$ for all $j$, since this is the case that we
are interested later on, and correspondingly we assume that $p_{j}\neq 0$
for all $j$. A very simple way to check that the state $\varphi $ is indeed
quite correlated, is to note that $\varphi (f\otimes f)>0$ for any non-zero
observable $f\in K$, by which we mean $f$ is self-adjoint in $K$, i.e. $f$
is real-valued. But it is easily seen that $\theta$ does not satisfy the
corresponding condition in the quantum case, namely if $A\in M_{n}$ is an
observable (i.e. it is self-adjoint) but all its diagonal entries are zero,
then $\theta (A\otimes A)=0$ even if $A$ is non-zero. This is despite the
fact that we do have $\theta (A\otimes A)\geq 0$ for all observables $A$. In
this sense $\theta$ is heuristically speaking not as correlated for the two
copies of the quantum system as $\varphi$ is for the two copies of the classical
system.

Heuristically, in order to have a quantum version of Eq. (\ref{klasfyn})
which is a good analogue of the classical situation, we need
to require the 2-system state in the quantum situation to be correlated in
the above sense for all observables, as is the case in the classical
situation, rather than just for some observables (namely for observables with
non-zero diagonal entries). But this then means that $\theta$ is not good
enough.

Exactly here entanglement comes to the rescue. Firstly, it is easily
verified that for the entangled state $\omega $ as defined in the previous
section we have $\omega (A\otimes A^{T})\geq 0$ for all observables $A$.
Note that this is not true for $\omega (A\otimes A)$, so this form is not
suitable for looking for correlations in the above sense. For $\theta$ we
have $\theta (A\otimes A^{T})=\theta (A\otimes A)$ so the two forms are
equivalent in the case of $\theta $. The form $\omega (A\otimes A^{T})$ is
the appropriate one to use in the case of $\omega $, and note that indeed
$\omega (A\otimes A^{T})>0$ for any non-zero observable $A\in M_{n}$, in
perfect analogy to the classical case. This suggests that it would be more
natural to use the entangled state $\omega $ in place of $\varphi $, rather
than the non-entangled state $\theta $, if we attempt to express quantum
detailed balance in the form of Eq. (\ref{klasfyn}) in terms of a state
which has a similar degree of correlations for the quantum observables that
$\varphi$ has for classical observables.

Below we use two definitions of quantum detailed balance appearing in the
literature to illustrate this connection with entanglement explicitly.

\section{Definitions of quantum detailed balance}

\label{a3}

We now describe two definitions of quantum detailed balance for which the
connection to the entangled state $\omega$ from Section \ref{a2} can be made
in a particularly clear way.

For a simple and clear discussion of how one can rewrite the classical
definition of detailed balance in a form that suggests the basic form of the
definitions of quantum detailed balance presented below, please refer to
\cite{FR, V}. This gives some intuition regarding the origins of these
definitions. Also see \cite{OSW, Ag, CWA} for some of the early literature on
detailed balance, as well as \cite{M3}. More specific references will be given
as we proceed.

As before we consider a system with $n$ dimensional Hilbert space. We allow
the system to interact with its environment, i.e. it is an open system. A
standard approach to this situation is to model the time-evolution of the
system in the Heisenberg picture as a quantum Markov semigroup (QMS) $\tau_{t}$
on the algebra $M_{n}$, where we take the time variable to be either
continuous, i.e. $t\geq0$, or discrete, i.e. $t=0,1,2,3,...$. This means that
for each $t$ the corresponding $\tau_{t}$ is a completely positive linear map
from $M_{n}$ to itself which is also unital, i.e. $\tau_{t}(I)=I$, and
furthermore the semigroup property $\tau_{s}\tau_{t}=\tau_{s+t}$ is satisfied.
Extensive discussions as to when a QMS is a good approximation to the physical
time-evolution is given for example in the books \cite{BP} and \cite{AL}, but
also see \cite{Da74} for one of the original papers.

It turns out that for the framework presented in this section and the results
discussed in the next, the semigroup property is not needed, so this
assumption can in fact be dropped, which may be relevant when studying
non-Markovian dynamics. We do however keep the rest of the above mentioned
assumptions regarding $\tau_{t}$, in which case we simply refer to $\tau_t$
as \emph{dynamics}. The literature on detailed balance related to our
approach typically assumes the semigroup property.

The first definition of quantum detailed balance we consider is from
\cite{MS}, and is called detailed balance II.
In \cite{MS} the dynamics is only assumed to be positive, rather than
completely positive, and they only consider the case of discrete time. We
therefore adapt their approach to completely positive maps and also to include
continuous time. Our results in the next section in fact still hold when
working with positivity instead of complete positivity, but as is well known
\cite{KGKS} there are convincing physical reasons to assume complete positivity,
and this also happens to be mathematically convenient in many cases.
In this regard also see again the books \cite{BP} and \cite{AL}.
The above mentioned extension from discrete to continuous time on the other hand
is a minor mathematical issue in our setup in this section.
All our arguments in this section, as well as Sections \ref{a4} and \ref{bewyse},
work for both the case of continuous time and the case of discrete time.

We are going to define detailed balance of the dynamics $\tau_{t}$ of the
system relative to a given fixed density matrix $\rho$ of the system. The key
mathematical idea to define and study detailed balance is to consider certain
duals or adjoints of $\tau_{t}$. In particular for detailed balance II we need
the following.

With $\left\langle \cdot\right\rangle $ the expectation functional given by
$\rho$ as in Section \ref{a2}, we can define the \emph{dual} (relative to
$\rho$) of any linear map $\alpha:M_{n}\rightarrow M_{n}$ as the linear map
$\alpha^{\prime}:M_{n}\rightarrow M_{n}$ such that
\[
\left\langle \alpha^{\prime}(A)B\right\rangle =
\left\langle A\alpha(B)\right\rangle
\]
for all $n\times n$ matrices $A$ and $B$. Note that since $\rho$ is
invertible, such an $\alpha^{\prime}$ necessarily exists and is unique, since
it can be obtained from the Hermitian adjoint of $\alpha$ with respect to the
inner product $(A,B)_{\rho}:=\tr(\rho A^{\dagger}B)=\left\langle
A^{\dagger}B\right\rangle $. Indeed, denoting this Hermitian adjoint by
$\alpha^{\rho}$, it is easy to check that
$\alpha^{\prime}(A)=\alpha^{\rho}(A^{\dagger})^{\dagger}$.

\begin{definition}
We say that $\tau_{t}$ as given above satisfies \emph{detailed balance II with
respect to} $\rho$ if $\tau_{t}^{\prime}$ is a completely positive unital
linear map for every $t$.
\end{definition}

As a general remark, note that if $\tau_{t}$ has the semigroup property, then
$\tau_{t}^{\prime}$ automatically has it as well, since
\[
\left\langle \tau_{s+t}^{\prime}(A)B\right\rangle
=\left\langle A\tau_{s+t}(B)\right\rangle
=\left\langle A\tau_{s}[\tau_{t}(B)]\right\rangle
=\left\langle \tau_{t}^{\prime}[\tau_{s}^{\prime}(A)]B\right\rangle .
\]

Note that roughly speaking detailed balance II boils down to requiring that the
dual $\tau_{t}^{\prime}$ is a sensible physical time-evolution.

Next we consider a type of standard quantum detailed balance (see\cite{DFFU},
and also \cite{PA} for related work). The particular form of standard quantum
detailed balance considered below was studied in \cite{FU2, FR2}. It will
immediately be seen that it is defined in a form directly related to the
entangled state $\omega$, a point we come back to in the next section. It is
defined in terms of a \textit{reversing operation}
$\Theta:M_{n}\rightarrow M_{n}$, meaning that $\Theta$ is a
$\ast$-anti-automorphism, i.e. it is linear,
$\Theta(A^{\dagger})=\Theta(A)^{\dagger}$ and $\Theta(AB)=\Theta(B)\Theta(A)$,
and we furthermore assume that $\Theta^{2}$ is the identity map on $M_{n}$.
Note that some form of time reversal plays a central role in a number of
approaches to detailed balance; see for example \cite{Ag, M2}, and also the
discussion in \cite{MS}.

For any linear $\alpha:M_{n}\rightarrow M_{n}$ we define its \textit{KMS-dual}
$\alpha^{(1/2)}:M_{n}\rightarrow M_{n}$ (relative to $\rho$) by
\[
\tr(\rho^{1/2}\alpha^{(1/2)}(A)\rho^{1/2}B)=\tr(\rho^{1/2}A\rho^{1/2}\alpha(B))
\]
for all $n\times n$ matrices $A$ and $B$. We note that $\alpha^{(1/2)}$ exists
and is uniquely determined. In fact it is easily seen to be given by
\[
\alpha^{(1/2)}(A)=\rho^{-1/2}\alpha^{\dagger}(\rho^{1/2}A^{\dagger}\rho
^{1/2})^{\dagger}\rho^{-1/2}
\]
where $\alpha^{\dagger}$ is the Hermitian adjoint of $\alpha$ with respect to
the Hilbert-Schmidt inner product. From this formula it also follows that
$\alpha^{(1/2)}$ is positive if $\alpha$ is, and completely positive if
$\alpha$ is. Furthermore, if $\tau_{t}$ is a QMS, it can be seen that
$\tau_{t}^{(1/2)}$ is as well. However, the semigroup property will again
not be essential for our work.

\begin{definition} We say that $\tau_{t}$ on $M_{n}$ satisfies \emph{standard
quantum detailed balance w.r.t. the reversing operation} $\Theta$ \emph{and the
density matrix} $\rho$, abbreviated as $\Theta$\emph{-sqdb w.r.t. $\rho$}, if
\[
\tau_{t}^{(1/2)}=\Theta\circ\tau_{t}\circ\Theta.
\]
\end{definition}

As the references above and in the introduction shows, there are also a number
of other definitions of quantum detailed balance in the literature. For remarks
comparing some of these definitions, we refer the reader to 
\cite{FU2, MS} in particular.

\section{Detailed balance and entanglement}

\label{a4}

In this section we turn to our main goal, namely to characterize quantum
detailed balance in terms of the entangled state $\omega$ introduced in
Section \ref{a2}. Here we only present the results along with some
discussion, while the technical details regarding their derivations are given
in the next section. As mentioned in Section \ref{a2}, $\rho$ is an
invertible density matrix throughout and we have chosen some fixed basis in
which $\rho$ is diagonal to define the transposition. Furthermore,  the term
\emph{dynamics} is as defined in the previous section.

The central tool towards our goal is the modular operator $\Delta$ defined by
\[
\Delta(A)=\rho A\rho^{-1}
\]
for all $n\times n$ matrices $A$. This operator is part of a very general
theory, namely modular theory or Tomita-Takesaki theory, which is discussed
for example in \cite{BR}, but since we work in finite dimensions we don't need
to delve into the general theory.

We start with the following characterization of detailed balance II in terms of
the modular operator.

\begin{theorem} \label{st1}
The dynamics $\tau_{t}$ satisfies detailed balance II w.r.t.
$\rho$ if and only if it commutes with the modular operator, i.e.
\begin{equation}
\tau_{t}\Delta=\Delta\tau_{t}, \label{kommuteer}
\end{equation}
and it leaves the state $\rho$ invariant in the sense that
\begin{equation}
\left\langle \tau_{t}(A)\right\rangle =\left\langle A\right\rangle
\label{invariant}
\end{equation}
for all $n\times n$ matrices $A$.
\end{theorem}

One direction of this theorem is given in \cite{MS}, namely that Eq.
(\ref{kommuteer}) and (\ref{invariant}) follow from detailed balance, but the
converse is not, though it is closely related to Theorem 6 of \cite{MS}. This
characterization of detailed balance II is one of the ingredients in deriving
the characterization of detailed balance II in terms of the entangled state
$\omega$ presented below.

For any linear map $\alpha:M_{n}\rightarrow M_{n}$ we can define another
linear map $\hat{\alpha}:M_{n}\rightarrow M_{n}$ by
\[
\hat{\alpha}(A)=\alpha^{\prime}(A^{\intercal})^{\intercal}
\]
where $\alpha^{\prime}$ is as defined in Section \ref{a3}. In order to
formulate the characterization of detailed balance II in terms of $\omega$, we
apply this to the dynamics $\tau_{t}$, i.e. we consider
$\hat{\tau}_{t}$ given by
\begin{equation}
\hat{\tau}_{t}(A)=\tau_{t}^{\prime}(A^{\intercal})^{\intercal}\label{hat}
\end{equation}
for all $n\times n$ matrices $A$ and every $t$. Keep in mind that
$\tau_{t}^{\prime}$ and therefore $\hat{\tau}_{t}$ are mathematically
well-defined operators for every $t$. However, it is only under the condition
of detailed balance II that $\tau_{t}^{\prime}$ becomes dynamics, i.e. that it
is unital and completely positive. When this is the case, $\hat{\tau}_{t}$
similarly becomes dynamics (see Section \ref{bewyse}). In certain examples,
similar to those in \cite{RFZ}, but in arbitrary finite dimensions, one can
show using Theorem \ref{st1} and Choi matrices \cite{Ch}
that the dynamics $\hat{\tau}_{t}$ is just the original dynamics $\tau_{t}$,
as opposed to $\tau_{t}^{\prime}$ which in such examples turns out to be in
effect a time-reversal of $\tau_{t}$. This cannot be expected to be true in
general though. Note that since the transpose appears in Eq. (\ref{hat}),
the definition of $\hat{\tau}_{t}$ is basis dependent, so we have made a
specific choice to fit in with our choice of $\omega$ from Section \ref{a2}.
When using the more general construction $\omega_{r}$, one could in principle
explore a corresponding generalization of Eq. (\ref{hat}), but here we deal
exclusively with Eq. (\ref{hat}).

Now we can characterize detailed balance II in terms of entanglement.

\begin{theorem} \label{st2}
The dynamics $\tau_{t}$ satisfies detailed balance II w.r.t.
$\rho$ if and only if
\begin{equation}
\omega\lbrack A\otimes\hat{\tau}_{t}(B)]=\omega\lbrack\tau_{t}(A)\otimes B]
\label{verst}
\end{equation}
for all $n\times n$ matrices $A$ and $B$, and
\begin{equation}
\hat{\tau}_{t}(I)=I , \label{unitaal}
\end{equation}
for every $t$. Alternatively Eq. (\ref{verst}) can be expressed as
\begin{equation}
\omega\circ(\operatorname{id}_{M_{n}}\otimes\hat{\tau}_{t})
=\omega\circ(\tau_{t}\otimes\operatorname{id}_{M_{n}}), \label{altverst}
\end{equation}
i.e. evolving the 2-system by $\operatorname{id}_{M_{n}}\otimes\hat{\tau}_{t}$
has the same effect on the entangled pure state $\omega$ as
$\tau_{t}\otimes\operatorname{id}_{M_{n}}$, where $\operatorname{id}_{M_{n}}$
denotes the identity map on the algebra $M_{n}$.
\end{theorem}

Next we consider a similar characterization of $\Theta$-sqdb. The
definition of $\Theta$-sqdb is indeed already in a form that is aligned with
$\omega$. We simply define $\alpha^{\Theta}:M_{n}\rightarrow M_{n}$ by
\[
\alpha^{\Theta}(A)=(\Theta\circ\alpha\circ\Theta(A^{\intercal}))^{\intercal}
\]
for any linear $\alpha:M_{n}\rightarrow M_{n}$. Then one can immediately
reformulate the definition of $\Theta$-sqdb to obtain the following
characterization which is inherent to the work in \cite{FR, FR2, BQ2}:

\begin{proposition} \label{prop1}
The dynamics $\tau_{t}$ satisfies $\Theta$-sqdb w.r.t.
$\rho$ if and only if
\[
\omega\lbrack A\otimes\tau_{t}^{\Theta}(B)]=\omega\lbrack\tau_{t}(A)\otimes B]
\]
for all $n\times n$ matrices $A$ and $B$ and every $t$.
\end{proposition}

A typical choice of $\Theta$ is $\Theta(A)=A^{\intercal}$. In this case
$\tau_{t}^{\Theta}=\tau_{t}$ and the above condition simplifies to
\[
\omega\lbrack A\otimes\tau_{t}(B)]=\omega\lbrack\tau_{t}(A)\otimes B]
\]
so this choice of $\Theta$ seems to fit in naturally with our choice of
$\omega$.

It is straightforward to construct examples of $\Theta$-sqdb in $M_{2}$ where
$\tau_{t}$ does not commute with $\Delta$, unlike the case of
detailed balance II. This aspect of standard quantum detailed balance was
emphasized in for example \cite{FU2}.

On the other hand, should we assume that $\tau_{t}$ does commute with $\Delta$,
one can show that $\Theta$-sqdb implies detailed balance II.

Lastly we mention that all of the results in this section still hold if we work
in terms of positivity instead of complete positivity, as discussed in the
previous section.

\section{Proofs} \label{bewyse}

Here we prove the results presented in Section \ref{a4}. We begin by discussing
a number of mathematical facts which will be of use in the proofs.

Given any linear map $\alpha:M_{n}\rightarrow M_{n}$ we define the linear map
$\alpha^{\ddagger}:M_{n}\rightarrow M_{n}$ by
\[
\tr[\alpha^{\ddagger}(A)B]=\tr[A\alpha(B)].
\]
Notice that it is a version of the dual $\alpha^{\prime}$, but w.r.t. the
trace instead of $\left\langle \cdot\right\rangle $. Similar to $\alpha
^{\prime}$, $\alpha^{\ddagger}$ can be obtained from the usual Hermitian
adjoint $\alpha^{\dagger}$ of the operator $\alpha$ with respect to the
Hilbert-Schmidt inner product by the formula
\[
\alpha^{\ddagger}(A)=\alpha^{\dagger}(A^{\dagger})^{\dagger}
\]
where the last $\dagger$ refers to the Hermitian adjoint of the $n\times n$
matrix $\alpha^{\dagger}(A^{\dagger})$. Note that $\alpha(I)=I$ if and only if
$\tr\circ\alpha^{\ddagger}=\tr$. It is similarly
easy to see that $\left\langle \alpha(A)\right\rangle =\left\langle
A\right\rangle $ for all $A$ if and only if $\alpha^{\ddagger}(\rho)=\rho$. In
the case that $\alpha$ is a Hermitian map, i.e. it satisfies $\alpha
(A^{\dagger})=\alpha(A)^{\dagger}$, we see that $\alpha^{\dagger}$ is also
Hermitian, since
\begin{align*}
\tr\left[  \alpha^{\dagger}(A^{\dagger})B\right]   &
=\tr\left[  A^{\dagger}\alpha(B)\right]
=\{\tr[\alpha(B^{\dagger})A]\}^{\ast}
=\{\tr[B^{\dagger}\alpha^{\dagger}(A)]\}^{\ast}\\  &
=\tr\left[  \alpha^{\dagger}(A)^{\dagger}B\right]  .
\end{align*}
Therefore
\[
\alpha^{\ddagger}=\alpha^{\dagger}
\]
if $\alpha$ is Hermitian. So in our physical context we in fact only need to
work with $\alpha^{\dagger}$, since positive maps are Hermitian.
Mathematically it will however be convenient to consider $\alpha^{\ddagger}$
as well.

From the definition $\Delta(A)=\rho A\rho^{-1}$ of the modular operator
$\Delta$, it is easily verified that $\Delta^{\dagger}=\Delta$, where again
the Hermitian adjoint $\Delta^{\dagger}$ is taken with respect to the
Hilbert-Schmidt inner product. I.e. $\Delta$ is self-adjoint, and similarly
$\Delta^{1/2}=\rho^{1/2}(\cdot)\rho^{-1/2}$ is self-adjoint. The latter means
that $\Delta=\Delta^{1/2}\Delta^{1/2}\geq0$. Furthermore,
$\Delta^{-1}=\rho^{-1}(\cdot)\rho$ exists so all of the eigenvalues of $\Delta$
are strictly positive, so in fact
\[
\Delta>0
\]
as an operator on the Hilbert space $M_{n}$ with the Hilbert-Schmidt norm.
This means that $\Delta^{-iz}$ is well-defined for all $z\in\mathbb{C}$. We
consider $\Delta^{-iz}$ rather than $\Delta^{z}$ as a convention, since then
in the case of a Gibbs state and real $z$ it follows that $\Delta^{-iz}$ is
essentially a scaled version of the system's isolated dynamics; see Eq.
(\ref{modgroep}) below.

A convenient and standard representation of a linear map
$\alpha :M_{n}\rightarrow M_{n}$, for example $\Delta$ above, is to arrange the
columns of an $n\times n$ matrix in order below one another in an $n^{2}$
dimensional column, in which case $\alpha$ can be written as an
$n^{2}\times n^{2}$ matrix. This is just a choice of basis, and is essentially
an explicit case of the GNS construction with respect to the trace (see for
example \cite{BR} for the general GNS construction). In this representation
$\alpha^{\dagger}$ is then easily seen to be represented by the Hermitian
adjoint of the $n^{2}\times n^{2}$ matrix (i.e. transpose and complex
conjugation).

Since we are working in a basis in which $\rho$ is diagonal, as mentioned in
Section \ref{a2}, namely Eq. (\ref{dm}), it follows that in the above mentioned
representation,
\begin{equation}
\Delta=\left[
\begin{array}
[c]{ccc}
\left[
\begin{array}
[c]{ccc}
\rho_{1}\rho_{1}^{-1} &  & \\
& \ddots & \\
&  & \rho_{n}\rho_{1}^{-1}
\end{array}
\right]  &  & \\
& \ddots & \\
&  & \left[
\begin{array}
[c]{ccc}
\rho_{1}\rho_{n}^{-1} &  & \\
& \ddots & \\
&  & \rho_{n}\rho_{n}^{-1}
\end{array}
\right]
\end{array}
\right]  \label{modop}
\end{equation}
where we have indicated $n\times n$ blocks for clarity. From this we see that
\begin{equation}
\Delta^{-iz}(A)=\rho^{-iz}A\rho^{iz}. \label{modgroep}
\end{equation}

Now we turn to the proofs of the results of the previous section. The first
step is the following:

Assuming that the dynamics $\tau_{t}$ of our system satisfies detailed
balance II w.r.t. $\rho$ as described in Section \ref{a3}, it follows that
\begin{equation}
\tau_{t}^{\prime}(A)=
\rho^{-1/2}\tau_{t}^{\dagger}(\rho^{1/2}A\rho^{1/2})\rho^{-1/2}
\label{formule}
\end{equation}
for all $n\times n$ matrices $A$, where $\tau_{t}^{\dagger}$ denotes the
Hermitian adjoint of $\tau_{t}$ with respect to the Hilbert-Schmidt inner
product.

The derivation of Eq. (\ref{formule}) is given in \cite{MS}, but we provide it
here for completeness in slightly more elementary
form, which is possible since we are working in finite dimensions. We in fact
prove something a bit more general than Eq. (\ref{formule}); see Eq.
(\ref{duaal2}). Along the way we prove some general results which will be used
in the subsequent proofs as well.

For any linear $\alpha:M_{n}\rightarrow M_{n}$ we have
$\left\langle\alpha^{\prime}(A)B\right\rangle
=\left\langle A\alpha(B)\right\rangle
=\tr[\alpha^{\ddagger}(\rho A)B]
=\left\langle \rho^{-1}\alpha^{\ddagger}(\rho A)B\right\rangle $,
therefore
\begin{equation}
\alpha^{\prime}(A)=\rho^{-1}\alpha^{\ddagger}(\rho A) \label{duaal}
\end{equation}
Furthermore,
\begin{align*}
\left\langle A\alpha(B)\right\rangle  &
=\left\langle \alpha^{\prime}(A)B\right\rangle
=\tr[B\rho\alpha^{\prime}(A)]
=\tr[\alpha^{\prime\ddagger}(B\rho)A]
=\tr[\rho A\alpha^{\prime\ddagger}(B\rho)\rho^{-1}]\\  &
=\left\langle A\alpha^{\prime\ddagger}(B\rho)\rho^{-1}\right\rangle
\end{align*}
so $\alpha(B)=\alpha^{\prime\ddagger}(B\rho)\rho^{-1}$, i.e. $\alpha
^{\prime\ddagger}(B\rho)=\alpha(B)\rho$. Assuming that $\alpha$ and
$\alpha^{\prime}$ are Hermitian, it follows that
$\alpha^{\ddagger}=\alpha^{\dagger}$ and
$\alpha^{\prime\ddagger}=\alpha^{\prime\dagger}$ are also Hermitian, therefore
we also have
$\alpha^{\prime\dagger}(\rho B)=\rho\alpha(B)$.
Hence
\begin{align*}
\left\langle A\alpha(B)\right\rangle  &
=\tr[\rho A\rho^{-1}\alpha^{\prime\dagger}(\rho B)]
=\tr[\alpha^{\prime}(\rho A\rho^{-1})\rho B]
=\left\langle \alpha^{\prime}(\rho A\rho^{-1})\rho B\rho^{-1}\right\rangle \\  &
=\left\langle A\rho^{-1}\alpha(\rho B\rho^{-1})\rho\right\rangle
\end{align*}
from which it follows that $\alpha(B)=\rho^{-1}\alpha(\rho B\rho^{-1})\rho$.

I.e. we have shown that
\begin{equation}
\alpha\Delta=\Delta\alpha\label{kommuteer2}
\end{equation}
if both $\alpha$ and $\alpha^{\prime}$ are Hermitian. But then it follows that
$\alpha\Delta^{-iz}=\Delta^{-iz}\alpha$, thinking in terms of operators on the
Hilbert space $M_{n}$, in other words
\begin{equation}
\alpha(\rho^{-iz}A\rho^{iz})=\rho^{-iz}\alpha(A)\rho^{iz} \label{kommuteer3}
\end{equation}
according to Eq. (\ref{modgroep}). This implies that
\begin{align*}
\tr[\alpha^{\dagger}(\rho^{iz}A\rho^{-iz})B]  &
=\tr[A\rho^{-iz}\alpha(B)\rho^{iz}]
=\tr[A\alpha(\rho^{-iz}B\rho^{iz})]\\  &
=\tr[\rho^{iz}\alpha^{\dagger}(A)\rho^{-iz}B]
\end{align*}
hence
$\alpha^{\dagger}(\rho^{iz}A\rho^{-iz})=\rho^{iz}\alpha^{\dagger}(A)\rho^{-iz}$.
In particular
$\rho^{1/2}\alpha^{\dagger}(A)\rho^{-1/2}=
\alpha^{\dagger}(\rho^{1/2}A\rho^{-1/2})$
so
$\rho^{1/2}\alpha^{\dagger}(\rho^{1/2}A\rho^{1/2})\rho^{-1/2}
=\alpha^{\dagger}(\rho A)=\rho\alpha^{\prime}(A)$,
where in the last equality we used Eq. (\ref{duaal}).

We have therefore shown that
\begin{equation}
\alpha^{\prime}(A)
=\rho^{-1/2}\alpha^{\dagger}(\rho^{1/2}A\rho^{1/2})\rho^{-1/2} \label{duaal2}
\end{equation}
for any Hermitian linear $\alpha:M_{n}\rightarrow M_{n}$ for which
$\alpha^{\prime}$ is also Hermitian.

In particular this means that Eq. (\ref{formule}) holds when $\tau_{t}$
satisfies detailed balance w.r.t. $\rho$, since then $\tau_{t}$ and
$\tau_{t}^{\prime}$ are both positive, and therefore Hermitian.

\emph{Proof of Theorem \ref{st1}.} Assume that $\tau_{t}$ satisfies detailed
balance II w.r.t. $\rho$. Then Eq. (\ref{kommuteer}) follows from
Eq. (\ref{kommuteer2}). Furthermore, Eq. (\ref{invariant}) holds, since
$\left\langle \tau_{t}(A)\right\rangle
=\left\langle \tau_{t}^{\prime}(I)A\right\rangle
=\left\langle A\right\rangle $
directly from the definition of $\tau_{t}^{\prime}$ and detailed balance II.

Now for the converse. First note that for a linear map
$\alpha:M_{n}\rightarrow M_{n}$
we have that $\alpha$ is completely positive if and only if
$\alpha^{\dagger}$ is completely positive. This follows immediately from the
definition of $\alpha^{\dagger}$ and the fact \cite{KGKS, S} that a linear map
$\varphi:M_{n}\rightarrow M_{n}$ is completely positive if and only if it can be
written in the form
\[
\varphi(A)=\sum_{j=1}^{n^{2}}V_{j}AV_{j}^{\dagger}
\]
for all $A$, for some set of matrices $V_{j}\in M_{n}$. (It can also be shown
by a slightly longer argument that $\alpha$ is positive if and only if
$\alpha^{\dagger}$ is positive.)

Assuming Eq. (\ref{kommuteer}) and Eq. (\ref{invariant}), we define
$\varphi_{t}:M_{n}\rightarrow M_{n}$ by
\[
\varphi_{t}(A)=\rho^{-1/2}\tau_{t}^{\dagger}(\rho^{1/2}A\rho^{1/2})\rho
^{-1/2},
\]
from which follows that
\begin{align*}
\left\langle \varphi_{t}(A)B\right\rangle  &
=\tr[\rho\rho^{-1/2}\tau_{t}^{\dagger}(\rho^{1/2}A\rho^{1/2})\rho^{-1/2}B]
=\tr[\rho^{1/2}A\rho^{1/2}\tau_{t}(\rho^{-1/2}B\rho^{1/2})]\\  &
=\left\langle A\rho^{1/2}\tau_{t}(\rho^{-1/2}B\rho^{1/2})\rho^{-1/2}\right\rangle
=\left\langle A\tau_{t}(B)\right\rangle
\end{align*}
where in the last step we applied
$\tau_{t}(\rho^{iz}A\rho^{-iz})=\rho^{iz}\tau_{t}(A)\rho^{-iz}$
which follows from Eq. (\ref{kommuteer}) just like
Eq. (\ref{kommuteer3}) followed from Eq. (\ref{kommuteer2}). This shows that
$\tau_{t}^{\prime}=\varphi_{t}$, i.e.
\begin{equation}
\tau_{t}^{\prime}(A)=
\rho^{-1/2}\tau_{t}^{\dagger}(\rho^{1/2}A\rho^{1/2})\rho^{-1/2},
\label{formule2}
\end{equation}
from which we conclude that $\tau_{t}^{\prime}$ is completely positive, since
$\tau_{t}$ and therefore $\tau_{t}^{\dagger}$ are. (Similarly,
$\tau_{t}^{\prime}$
is positive if we only assume that $\tau_{t}$ is positive.)
Furthermore
\[
\left\langle \tau_{t}^{\prime}(I)A\right\rangle
=\left\langle \tau_{t}(A)\right\rangle
=\left\langle A\right\rangle ,
\]
implying that $\tau_{t}^{\prime}$ is unital. This shows that $\tau_{t}$
satisfies detailed balance II w.r.t. $\rho$ as required.

\emph{Proof of Theorem \ref{st2}.} Assume that $\tau_{t}$ satisfies detailed
balance II w.r.t. $\rho$. Then Eq. (\ref{formule}) holds as already shown
above, so by also using Eq. (\ref{omega}) and Eq. (\ref{hat}) it follows that
\begin{align*}
\omega\lbrack A\otimes\hat{\tau}_{t}(B)]  &
=\tr[\rho^{1/2}A\rho^{1/2}\tau_{t}^{\prime}(B^{\intercal})]\\  &
=\tr[\rho^{1/2}A\rho^{1/2}
 \rho^{-1/2}\tau_{t}^{\dagger}(\rho^{1/2}B^{\intercal}\rho^{1/2})\rho^{-1/2}]\\  &
=\tr[\tau_{t}(A)\rho^{1/2}B^{\intercal}\rho^{1/2}]
=\omega\lbrack\tau_{t}(A)\otimes B],
\end{align*}
i.e. Eq. (\ref{verst}) holds. Since $\tau_{t}^{\prime}(I)=I$ because of
detailed balance II, we also have Eq. (\ref{unitaal}) by Eq. (\ref{hat}).

Conversely, assuming Eqs. (\ref{verst}) and (\ref{unitaal}), we are going to
use Theorem \ref{st1} to show that $\tau_{t}$ satisfies detailed balance II
w.r.t. $\rho$. Since
\[
\omega(A\otimes B)
=\tr(\rho B^{\intercal}\rho^{1/2}A\rho^{-1/2})
=\left\langle B^{\intercal}\Delta^{1/2}(A)\right\rangle ,
\]
we have by our assumption Eq. (\ref{verst}) that
\begin{align*}
\left\langle B^{\intercal}\tau_{t}[\Delta^{1/2}(A)]\right\rangle  &
=\left\langle \tau_{t}^{\prime}(B^{\intercal})\Delta^{1/2}(A)\right\rangle
=\omega\lbrack A\otimes\tau_{t}^{\prime}(B^{\intercal})^{\intercal}]
=\omega\lbrack A\otimes\hat{\tau}_{t}(B)]\\  &
=\omega\lbrack\tau_{t}(A)\otimes B]
=\left\langle B^{\intercal}\Delta^{1/2}[\tau_{t}(A)]\right\rangle
\end{align*}
which means that $\tau_{t}\Delta^{1/2}=\Delta^{1/2}\tau_{t}$, hence
$\tau_{t}\Delta=\Delta\tau_{t}$. Furthermore,
\[
\left\langle \tau_{t}(A)\right\rangle
=\omega\lbrack\tau_{t}(A)\otimes I]
=\omega\lbrack A\otimes\hat{\tau}_{t}(I)]
=\left\langle A\right\rangle ,
\]
since we assumed that $\hat{\tau}_{t}(I)=I$. The conditions in Theorem \ref{st1}
are therefore satisfied, implying that $\tau_{t}$ satisfies detailed balance II
w.r.t. $\rho$, completing the proof of Theorem \ref{st2}.

\emph{Remarks regarding }$\hat{\tau}_{t}$\emph{ as dynamics.} Note
that for a linear map $\alpha:M_{n}\rightarrow M_{n}$ we have that $\alpha$ is
completely positive if and only if $\bar{\alpha}$ is completely positive, where
$\bar{\alpha}$ is defined by $\bar{\alpha}(A)=\alpha(A^{\intercal})^{\intercal}$
in terms of the transposition in our chosen basis as discussed in Section \ref{a2}.
This again follows from the representation of completely positive maps used in
the proof of Theorem \ref{st1}. In particular it then follows from
$\hat{\tau}_{t}(A)=\tau_{t}^{\prime}(A^{\intercal})^{\intercal}$ that
$\hat{\tau}_{t}$ is completely positive if $\tau_{t}^{\prime}$ is.
(Since transposition is a positive map, the corresponding results in terms of
positivity instead of complete positivity also hold.)
Clearly $\hat{\tau}_{t}$ is unital if $\tau_{t}^{\prime}$ is. Should we work
with the case where $\tau_{t}$ has the semigroup property,
then $\tau_{t}^{\prime}$ has the semigroup property as well, as explained in
Section \ref{a3}, from which it is easily seen that $\hat{\tau}_{t}$ also has
the semigroup property.

\section{Thermofield dynamics}

\label{tvd}

The characterization of detailed balance in terms of the entangled state $\omega$
turns out to fit naturally into the framework of thermofield dynamics and in this
section our goal is to show this. Our first step is to briefly outline some
of the basic elements of thermofield dynamics in a finite dimensional set-up.

Thermofield dynamics was developed in \cite{TU}, although a number of the key
ideas already appeared in \cite{AW, HHW}. A very useful discussion of
thermofield dynamics can be found in \cite{U}. The formulation in terms of
operator algebras was presented in \cite{O}, and reviewed in \cite{LV}. Our
exposition is largely based on the latter two sources, but adapted to our
setting.

The basic idea is to double the degrees of freedom of the system in the sense
that for each element $A$ of the system's observable algebra $M_{n}$ we define
an element $\tilde{A}$ of the commutant of $M_{n}$ in a cyclic representation
given by the GNS construction for the faithful state
$\left\langle \cdot\right\rangle$ on $M_{n}$ given by $\rho$. This element has
to satisfy a basic identity of thermofield dynamics called the tilde substitution
rule, namely
\[
\Delta^{-1/2}(\tilde{A}\rho^{1/2})=A^{\dagger}\rho^{1/2}
\]
for all $A\in M_{n}$. We need to find $\tilde{A}$ explicitly in a convenient
representation. There are different, though unitarily equivalent, ways of
writing the cyclic representation. For our purposes in this section it is most
convenient to first represent $M_{n}$ by $M_{n}\otimes I$, as a subalgebra of
$M_{n}\otimes M_{n}$, in which case its commutant is given by $I\otimes M_{n}$.
Furthermore, using our faithful representation $\pi$ of $M_{n}\otimes M_{n}$
from Section \ref{a2}, we obtain the cyclic representation of $M_{n}$
we are going to use, namely $A\mapsto\pi(A\otimes I)$, the cyclic vector being
$\rho^{1/2}$ in the Hilbert space $M_{n}$ with Hilbert-Schmidt norm. Note that
$(\rho^{1/2}|\pi(A\otimes I)\rho^{1/2})=\left\langle A\right\rangle$ as is
required of a cyclic representation associated to
$\left\langle \cdot\right\rangle$. It is then a simple matter to verify that
the tilde substitution rule above is satisfied exactly when we set
\[
\tilde{A}=\pi(I\otimes\bar{A})
\]
for all $A\in M_{n}$, where $\bar{A}$ is the complex conjugate of $A$, i.e.
each entry of $A$ is replaced by its complex conjugate. Indeed, we then have
\[
\Delta^{-1/2}(\tilde{A}\rho^{1/2})
=\rho^{-1/2}(\rho^{1/2}\bar{A}^{\intercal})\rho^{1/2}
=A^{\dagger}\rho^{1/2}
\]
as required.

It is also clear from the latter that the tilde substitution rule is in fact
simply an alternative way to write the definition of the modular operator
$\Delta$. Moreover, as one might expect from the fact that $\tilde{A}$ lies in
the commutant, it can alternatively be obtained from the modular conjugation of
Tomita-Takesaki theory \cite{O, LV}. Therefore thermofield dynamics is in a
sense contained in Tomita-Takesaki theory.

From a more physical point of view one can keep in mind that the KMS condition
can be written as
\[
\left\langle A\Delta(B)\right\rangle =\left\langle BA\right\rangle
\]
for all $A,B\in M_{n}$, and this is yet another way of writing the definition
of $\Delta$. So the tilde substitution rule is in effect simply a way to write
the KMS condition, i.e. to express thermal equilibrium.

Another core aspect of thermofield dynamics is the fact that
\[
\left\langle A\right\rangle=\omega(A\otimes I)
\]
as in Section \ref{a2}, i.e. expectation values for the mixed state
$\left\langle \cdot\right\rangle $, that is to say $\rho$, can be expressed in
terms of the pure state $\omega$.

This summarizes the main points from thermofield dynamics that are relevant
for us. Further background, motivation and applications can be found in the
references mentioned above. We now proceed to study detailed balance in this
framework. To do this, it is convenient to extend the definition in Section
\ref{a2} of the expectation functional $\left\langle \cdot\right\rangle$ to the
algebra $\pi(M_n\otimes M_n)$ in the following way that fits in neatly with the
thermofield dynamics framework:
\[
\left\langle A\tilde{B}\right\rangle =(\rho^{1/2}|A\tilde{B}\rho^{1/2})
\]
for all $A,B\in M_{n}$, where we have written $A$ as
shorthand for $\pi(A\otimes I)$, which is natural, since $\pi(A\otimes I)X=AX$
for any $X\in M_{n}$. The point of this is that it can also be rewritten as
\[
\left\langle A\tilde{B}\right\rangle =\omega(A\otimes\bar{B})
\]
which will allow us to write our entanglement characterizations of detailed
balance from Section \ref{a4} easily in the framework of thermofield dynamics.

We now have the following:

\begin{theorem}
Consider dynamics $\tau_{t}$ as described in Section \ref{a3}.

(a) The dynamics $\tau_{t}$ satisfies detailed balance II w.r.t. $\rho$ if and
only if
\begin{equation}
\left\langle \tau_{t}(A)\tilde{B}\right\rangle =
\left\langle A\widetilde{\tau_{t}^{\prime}(B)}\right\rangle \label{tvdkar}
\end{equation}
for all $A,B\in M_{n}$, and
\[
\tau_{t}^{\prime}(I)=I
\]
for every $t$.

(b) The dynamics $\tau_{t}$ satisfies $\Theta$-sqdb w.r.t. $\rho$ if and only
if
\[
\left\langle \tau_{t}(A)\tilde{B}\right\rangle =
\left\langle A[\Theta\circ\tau_{t}\circ\Theta(B)]^{\symbol{126}}\right\rangle
\]
for all $A,B\in M_{n}$, where $[\cdot]^{\symbol{126}}$ means we apply the
tilde to the contents of $[\cdot]$.
\end{theorem}

\begin{proof}
Note that Eq. (\ref{tvdkar}) is equivalent to $\omega(\tau_{t}(A)\otimes
\bar{B})=\omega(A\otimes\overline{\tau_{t}^{\prime}(B)})$. Taking the complex
conjugate of this, we see that it is in turn equivalent to $\omega
(\tau_{t}(A^{\dagger})\otimes B^{\intercal})=\omega(A^{\dagger}\otimes\tau
_{t}^{\prime}(B)^{\intercal})$, since $\tau_{t}(A)^{\dagger}=\tau
_{t}(A^{\dagger})$. So we have shown that Eq. (\ref{tvdkar}) is equivalent to
Eq. (\ref{verst}). The rest of the proof of this theorem is now
straightforward from the results of Section \ref{a4}.
\end{proof}

This theorem shows that the entanglement characterizations of detailed balance 
in Section \ref{a4} fit naturally into the framework of thermofield dynamics.

\section{Discussion}

\label{a6}

One may ask what the most fruitful ways are
to motivate formulations of quantum detailed balance on direct
physical grounds. Possibly characterization of detailed balance in terms
of an entangled state can provide an alternative quantum mechanical foundation
for, and interpretation of, detailed balance in terms of entanglement.

This would be in line with recent work where foundational aspects of
statistical mechanics are studied and motivated directly in terms of
entanglement \cite{GL}. This approach to the
foundations of statistical mechanics appears promising, so despite the more
traditional arguments in favour of the various quantum formulations of
detailed balance, attempting to motivate it from the perspective of
entanglement may prove fruitful. Possibly it could also give a wider
perspective on detailed balance, considering that here we only used a very
specific entangled state, while in principle one could consider conditions as
in Theorem \ref{st2} and Proposition \ref{prop1} with respect to more general 
entangled states. We hope that the connection between detailed balance and 
entanglement considered in this paper can further such studies.

\section*{Acknowledgement}

This research was supported by the National Research Foundation of South Africa.
We thank the referees as well as W. A. Majewski for suggestions to improve the
original version of the paper.

\end{document}